\documentclass{sig-alternate-10pt}
\input{Definitions}
\usepackage{amssymb,latexsym,amsmath}
\usepackage{algorithm}
\usepackage{algpseudocode}
\usepackage{array}
\usepackage{bbm}
\usepackage{url}
\usepackage{graphicx}
\usepackage{subfigure}
\usepackage{color}
\begin{document}
\numberofauthors{2}
 \author{
 \alignauthor
 Arvind Arasu\\
   \affaddr{Microsoft Research}\\
   \email{arvinda@microsoft.com}
 \alignauthor
 Raghav Kaushik\\
   \affaddr{Microsoft Research}\\
   \email{skaushi@microsoft.com}
}

\date{}
\title{Oblivious Query Processing}

\newcommand{\domain}[1]{\ensuremath{\mathcal{D}(#1)}}
\newcommand{\select}{\ensuremath{\sigma}}
\newcommand{\project}{\ensuremath{\pi}}
\newcommand{\join}{\ensuremath{\bowtie}}
\newcommand{\cproduct}{\ensuremath{\times}}
\newcommand{\union}{\ensuremath{\cup}}
\newcommand{\gunion}{\ensuremath{\bar{\cup}}}
\newcommand{\isect}{\ensuremath{\cap}}
\newcommand{\groupagg}{\ensuremath{\mathbbm{G}}}

\newcommand{\expand}{\ensuremath{\mathrm{Exp}}}
\newcommand{\idfunc}{\ensuremath{\mathrm{ID}}}
\newcommand{\rsfunc}{\ensuremath{\mathrm{RSum}}}
\newcommand{\attrs}{\ensuremath{\mathit{Attr}}}

\algblockdefx[CFor]{Foreach}{EndForeach}
[2]{\textbf{for all} #1 \textbf{in} #2 \textbf{do}} {\textbf{endfor}}

\algblockdefx[CFor]{ForIter}{EndForIter}
[2]{\textbf{for} #1 \textbf{to} #2 \textbf{do}} {\textbf{endfor}}

\maketitle

\begin{abstract}

  Motivated by cloud security concerns, there is an increasing
  interest in database systems that can store and support queries over
  encrypted data. A common architecture for such systems is to use a
  \emph{trusted} component such as a cryptographic co-processor for
  query processing that is used to securely decrypt data and perform
  computations in plaintext. The trusted component has limited memory,
  so most of the (input and intermediate) data is kept encrypted in an
  \emph{untrusted} storage and moved to the trusted component on
  ``demand.''

  In this setting, even with strong encryption, the data access
  pattern from untrusted storage has the potential to reveal sensitive
  information; indeed, all existing systems that use a trusted
  component for query processing over encrypted data have this
  vulnerability. In this paper, we undertake the first formal study of
  \emph{secure query processing}, where an adversary having full
  knowledge of the query (text) and observing the query execution
  learns nothing about the underlying database other than the result
  size of the query on the database. We introduce a simpler notion,
  \emph{oblivious query processing}, and show formally that a query
  admits secure query processing \emph{iff} it admits oblivious query
  processing. We present oblivious query processing algorithms for a
  rich class of database queries involving selections, joins, grouping
  and aggregation. For queries not handled by our algorithms, we
  provide some initial evidence that designing oblivious (and
  therefore secure) algorithms would be hard via reductions from two
  simple, well-studied problems that are generally believed to be
  hard. Our study of oblivious query processing also reveals
  interesting connections to\linebreak database join theory.
   
\end{abstract}


\section{Introduction}
\label{sec:intro}

There is a trend towards moving database functionality to the cloud
and many cloud providers have a \emph{database-as-a-service}
\emph{(DbaaS)} offering~\cite{amazonrds, sqlazure}. A DbaaS allows an
application to store its database in the cloud and run queries over
it. Moving a database to the cloud, while providing well-documented
advantages \cite{CurinoJPMWMBZ11}, introduces data security
concerns~\cite{smesurvey}. Any data stored on a cloud machine is
potentially accessible to snooping administrators and to attackers who
gain illegal access to cloud systems. There have been well-known
instances of data security breaches arising from such
adversaries~\cite{wsj}.

A simple mechanism to address these security concerns is
encryption. By keeping data stored in the cloud encrypted we can
thwart the kinds of attacks mentioned above. However encryption makes
computation and in particular \emph{query processing} over data
difficult. Standard encryption schemes are designed to ``hide'' data
while we need to ``see'' data to perform computations over
it. Addressing these challenges and designing database systems that
support query processing over encrypted data is an active area of
research~\cite{cipherbase, BajajS11, HacigumusILM02, PopaRZB11,
  monomi} and industry effort~\cite{sqlenc, oracleenc}.

A common architecture~\cite{cipherbase, BajajS11} for query processing
over encrypted data involves using \emph{trusted hardware} such as a
cryptographic co-processor~\cite{ibmcards}, designed to be
inaccessible to an adversary. The trusted hardware has access to the
encryption key, some computational capabilities, and limited
storage. During query processing, encrypted data is moved to the
trusted hardware, decrypted, and computed on. The trusted hardware has
limited storage so it is infeasible to store within it the entire
input database or intermediate results generated during query
processing; these are typically stored encrypted in an
\emph{untrusted} storage and moved to the trusted hardware only when
necessary. Other approaches to query processing over encrypted data
rely on (partial) \emph{homomorphic encryption}~\cite{HacigumusILM02,
  PopaRZB11, monomi} or using the client as the trusted
module\footnote{All of our results hold for this setting, but they are
  less interesting.}~\cite{HacigumusILM02, monomi}. These approaches
have limitations in terms of the class of queries they can handle or
data shipping costs they incur (see Section~\ref{sec:rel}), and they
are not the main focus of this paper.

The systems that use trusted hardware currently provide an
\emph{operational} security guarantee that any data outside of trusted
hardware is encrypted~\cite{cipherbase, BajajS11}. However, this
operational guarantee does not translate to \emph{end-to-end data
  security} since, even with strong encryption, the data movement
patterns to and from trusted hardware can potentially reveal
information about the underlying data. We call such information
leakage \emph{dynamic information leakage} and we illustrate it using
a simple join example.
\begin{example}
\label{ex1}
Consider a database with tables \texttt{\small Patient}
\texttt{\small (PatId,Name,City)} and \texttt{\small
  Visit(PatId,Date,Doctor)} storing patient details and their doctor
visits. These tables are encrypted by encrypting each record using a
standard encryption scheme and stored in untrusted memory. Using
suitably strong encryption\footnote{And padding to mask record
  sizes.}, we can ensure that the adversary does not learn anything
from the encrypted tables other than their sizes. (Such an encryption
scheme is non-deterministic so two encryptions of the same record
would look seemingly unrelated.) 

Consider a query that joins these two tables on \texttt{\small PatId}
column using a nested loop join algorithm. The algorithm moves each
patient record to the trusted hardware where it is decrypted. For each
patient record, all records of \texttt{\small Visit} table are moved
one after the other to the trusted hardware and decrypted. Whenever
the current patient record $p$ and visit record $v$ have the same
\texttt{\small PatId} value, the join record $\langle p, v \rangle$ is
encrypted and produced as output. An adversary observing the sequence
of records being moved in and out of trusted hardware learns the
join-graph. For example, if $5$ output records are produced in the
time interval between the first and second \texttt{\small Patient}
record moving to the trusted hardware, the adversary learns that some
patient had $5$ doctor visits.
\end{example}

The above discussion raises the natural question\linebreak whether we
can design query processing algorithms that avoid such dynamic
information leakage and provide end-to-end data security. The focus of
this paper is to seek an answer to this question; in particular, as a
contribution of this paper, we formalize a strong notion of
(end-to-end) secure query processing, develop efficient and secure
query processing algorithms for a large class of queries, and discuss
why queries outside of this class are unlikely to have efficient
secure algorithms.

There exists an extensive body of work \emph{Oblivious}\linebreak \emph{RAM}
\emph{(ORAM) Simulation} \cite{DBLP:journals/jacm/GoldreichO96,
  DBLP:conf/codaspy/GoodrichMOT12,pathoram,
  DBLP:conf/ndss/WilliamsS08}, a general technique that makes memory
accesses of an arbitrary program appear random by continuously
shuffling memory and adding spurious accesses. In Example~\ref{ex1},
with ORAM simulation the data accesses would appear random to an
adversary and we can show that the adversary learns no information
other than the total number of data accesses. 

\emph{Given general ORAM simulation, why design specialized secure
  query processing algorithms?}  We defer a full discussion of this
issue to Section~\ref{sec:oram}, but for a brief motivation consider
sorting an encrypted array of size $n$.  Just as in Example~\ref{ex1},
the data access patterns of a standard sorting algorithm such as
quicksort reveals information about the underlying data. An ORAM
simulation of quicksort would hide the access patterns; indeed, the
adversary does not even learn that a sort operation is being
performed. However, with current state-of-the-art ORAM algorithms, it
would incur an overhead of $\Theta(\log^2 n)$ per
access\footnote{Assuming ``small'' $\mathrm{polylog}(n)$ trusted
  memory.} of the original algorithm making the overall complexity of
sorting $\Theta (n \log^3 n)$. Instead, we could exploit the semantics
of sorting and design a secure sorting algorithm that has the
(optimal) time complexity of $\Theta (n \log
n)$~\cite{DBLP:journals/jacm/Goodrich11}. Here the adversary does
learn that the operation being performed is sorting (but does not
learn anything about the input being sorted) but we get significant
performance benefits. As we show in the rest of the paper, designing
specialized secure query processing algorithms helps us gain similar
performance advantages over generic ORAM simulations.

The exploration in this paper is part of the \emph{Cipherbase}
project~\cite{cbase-proj}, a larger effort to design and prototype a
comprehensive database system, relying on specialized hardware for
storing and processing encrypted data in the cloud.

\subsection{Overview of Contributions}
\label{sec:oview}
\vspace{1ex} \noindent \textbf{Secure Query Processing: } Informally,
we define a query processing algorithm for a query $Q$ to be
\emph{secure} if an adversary having full knowledge of the text of $Q$
and observing the execution of $Q$ does not learn anything about the
underlying database $D$ other than the result size of $Q$ over
$D$. The query execution happens within a \emph{trusted module (TM)}
not accessible to the adversary. The input database and possibly
intermediate results generated during query execution are stored
(encrypted) in an \emph{untrusted memory (UM)}. The adversary has
access to the untrusted memory and in particular can observe the
sequence of memory locations accessed and data values read and written
during query execution. Throughout, we assume a \emph{passive}
adversary, who does not actively interfere with query processing.

When defining security, we grant the adversary knowledge of query $Q$
which ensures that data security does not depend on the query being
kept secret. Databases are typically accessed through applications and
it is often easy to guess the query from a knowledge of the
application. Our formal definition of secure query processing
(Section~\ref{sec:sqp-def}) generalizes the informal definition above
and incorporates query security in addition to data security. Our
formal definition relies on machinery from standard cryptography such
as \emph{indistinguishability experiments}, but there are some
subtleties specific to database systems and applications that we
capture; Appendix~\ref{sec:query-sec} discusses these issues in
greater detail.

We note that simply communicating the (encrypted) query result over an
untrusted network reveals the result size, so a stronger notion of
security seems impractical in a cloud setting. Also, our definition of
secure query processing implies that an adversary with an access to
the cloud server gets no advantage over an adversary who can observe
only the communications between the client and the database server,
assuming both of them have knowledge of the query.

\vspace{1ex}\noindent \textbf{Oblivious Query Processing Algorithms: }
Central to the idea of secure query processing is the notion of
\emph{oblivious query processing}. Informally, a query processing
algorithm is oblivious if its (untrusted) memory access pattern is
independent of the database contents once we fix the query and its
input and output sizes. We can easily argue that any secure algorithm
is oblivious: otherwise, the algorithm has different memory access
patterns for different database instances, so the adversary learns
something about the database instance by observing the memory access
pattern of the algorithm. Interestingly, obliviousness is also a
\emph{sufficient} condition in the sense that any oblivious algorithm
can be made secure using standard cryptography
(Theorem~\ref{thm:obl}). The idea of reducing security to memory
access obliviousness was originally proposed
in~\cite{DBLP:journals/jacm/GoldreichO96} for general programs in the
context of software protection. Note that obliviousness is defined
with respect to memory accesses to the untrusted store; any memory
accesses internal to TM are invisible to the adversary and do not
affect security.

Our challenge therefore is to design oblivious algorithms for database
queries. We seek oblivious algorithms that have small TM memory
footprint since all practical realizations of the trusted module such
as cryptographic co-processors have limited storage (few
MBs) \cite{ibmcards}. Without this restriction a simple oblivious
algorithm is to read the entire database into TM and perform query
processing completely within TM.

To illustrate challenges in designing oblivious query processing
algorithms consider the simple join query $R(A,\ldots)$ $\join
S(A,\ldots)$ which seeks all pairs of tuples from $R$ and $S$ that
agree on attribute $A$. Figure~\ref{fig:joininst} shows two instances
for this join, represented as a binary ``join-graph''. Each $R$ and
$S$ tuple is shown as a vertex and its attribute $A$ value is shown
adjacent to it (lower case letters $a, b, \ldots$). An edge exists
between an $R$ tuple and an $S$ tuple having the same value in
attribute $A$, and each edge represents a join output. We note that
both instances have the same input output characteristics, $|R| = |S|
= |R \join S| = 16$, so an oblivious algorithm is required to have the
same memory access pattern for both instances. However, the internal
structure of the join graph is greatly different. All ``natural'' join
algorithms that use sorting or hashing to bring together joinable
tuples are sensitive to the join graph structure and therefore not
oblivious.

\begin{figure}
\centering
\subfigure[Instance 1] {
\fbox{
\includegraphics[trim=48ex 42ex 49ex 55ex, width=1.3in]{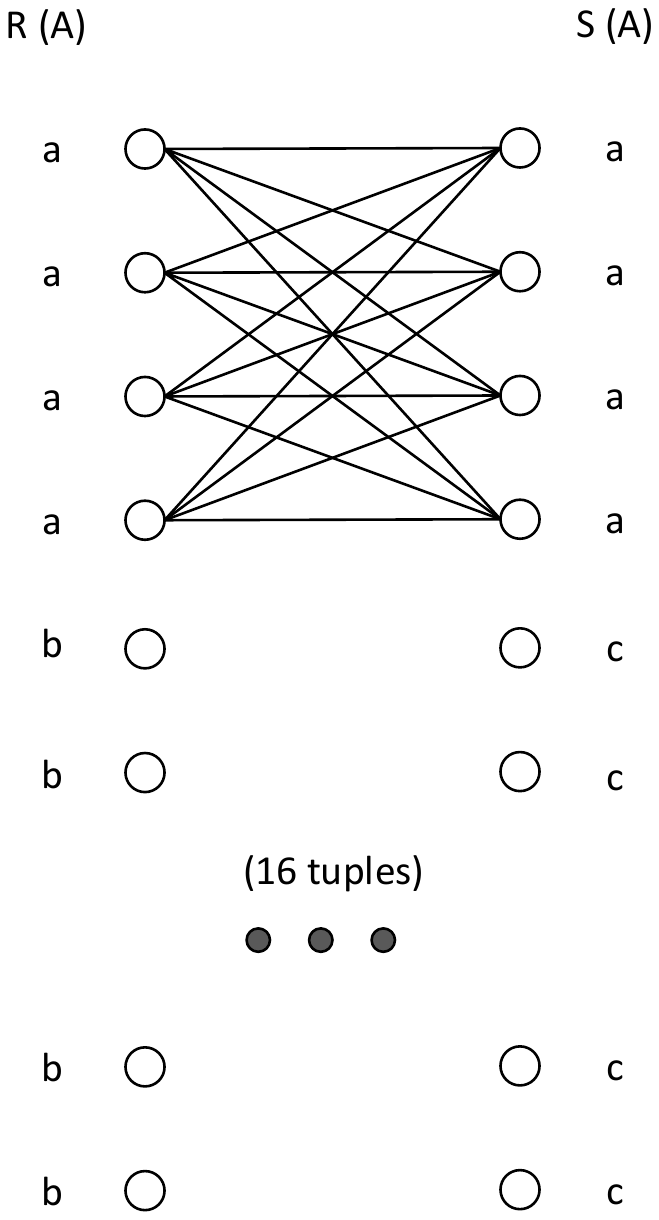}
}
\label{fig:joininsta}
}
\subfigure[Instance 2] {
\fbox{
\includegraphics[trim=48ex 42ex 49ex 55ex, width=1.3in]{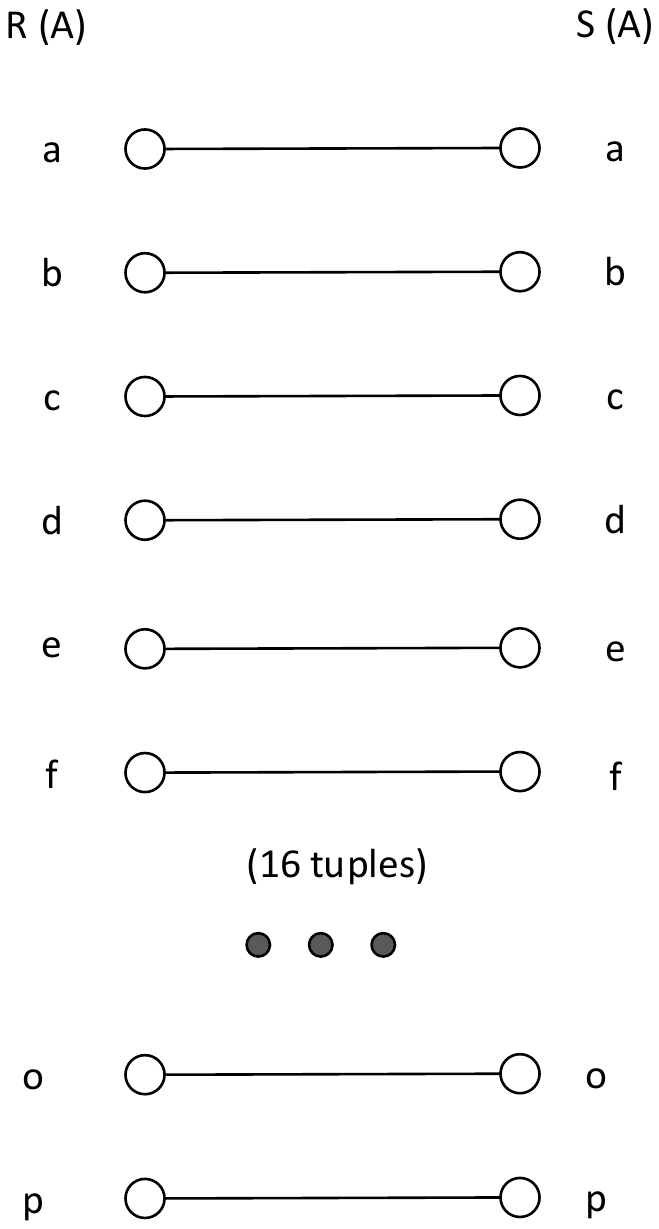}
}
\label{fig:joininstb}
}
\caption{Two join instances with same input/output sizes}
\label{fig:joininst}
\end{figure}

Traditional database query processing is \emph{non-oblivious} for two
reasons: First, traditional query processing proceeds by identifying a
query plan, which is a tree of operators with input tables at the
leaves. The operators are (conceptually) evaluated in a bottom-up
fashion and the output of each operator forms an input of its
parent. In some cases, this bottom-up evaluation can be pipelined.  In
others, the output of an operator needs to be generated fully before
the parent can consume it, and such intermediate output needs to be
temporarily stored in untrusted memory. This renders the overall query
processing non-oblivious since the size of the intermediate output can
vary depending on the database instance, even if we fix input and
output sizes.  Second, standard implementations of database operators
such as filters, joins, and grouping are not oblivious, so even if the
query plan consisted of a single operator, the resulting query
processing algorithm would not be oblivious.  In summary, traditional
query processing is non-oblivious at both inter- and intra-operator
levels, and we need to fundamentally rethink query processing to make
it oblivious.

Our first main algorithmic contribution is that a surprisingly rich
class of database queries admit {\em efficient} oblivious algorithms
(Sections~\ref{sec:overview},~\ref{sec:prim} and~\ref{sec:algorithm}).
\begin{theorem} 
\label{thm:informal}
(Informal) There exists an oblivious (secure) query processing
algorithm that requires $O(\log n)$ storage in TM for any database
query involving joins, grouping aggregation (as the outermost
operation), and filters, if (1) the non-foreign key join graph is
acyclic and (2) all grouping attributes are connected through foreign
key joins, where $n$ denotes the sum of query input and output
sizes. Further, assuming no auxiliary structures, the running time of
the algorithm is within $O(\log n)$ (multiplicative factor) of the
running time of the best insecure algorithm.
\end{theorem}
Theorem~\ref{thm:informal} suggests an interesting connection between
secure query processing and database join theory since \emph{acyclic
  joins} are a class of join queries known to be
tractable~\cite{DBLP:conf/pods/PaghP06}. We note that the class of
queries is fairly broad and representative of real-world analytical
queries. For example, most queries in the well-known TPC-H
benchmark~\cite{tpch} belong to this class, i.e., admit efficient
secure algorithms.

Assuming no auxiliary structures such as indexes, our algorithms are
efficient and within $O(\log n)$ of the running time of the best
insecure algorithm. While the no-index condition makes our results
less relevant for transactional workloads, where indexes play an
important role, they are quite relevant for analytical workloads where
indexes play a less critical role. In fact, query processing in
emerging database architectures such as \emph{column stores}
\cite{DBLP:journals/pvldb/AbadiBH09} has limited or no dependence on
indexes.

Further, with minor modifications to our algorithms using the
oblivious external memory sort algorithm
of~\cite{DBLP:conf/spaa/Goodrich11}, we get oblivious and secure
algorithms with excellent external (untrusted) memory characteristics.
\begin{theorem}
\label{thm:ext-informal}
(Informal) For the class of queries in Theorem~\ref{thm:informal}
there exists an oblivious (secure) algorithm with I/O complexity
within multiplicative factor\linebreak $\log_{M/B} (n/B)$ of that of the optimal
insecure algorithm, where $B$ is the block size and $M$ is the TM
memory.
\end{theorem}
In particular, if we have $\Omega (\sqrt{n})$ memory in TM our
external memory algorithms perform a constant number of scans to
evaluate the queries they handle.

Interestingly, for the special case of joins, secure algorithms have
been studied in the context of privacy preserving data integration
\cite{DBLP:conf/icde/LiC08}. The algorithm proposed in
\cite{DBLP:conf/icde/LiC08} proceeds by computing a cross product of
the input relations followed by a (secure) filter. Our algorithms are
significantly more efficient and handle grouping and aggregation.

\vspace{1ex}\noindent\textbf{Negative Results:} We have reason to
believe that queries outside of the class specified in
Theorem~\ref{thm:informal} do not admit secure efficient
algorithms. We show that the existence of secure algorithms would
imply more efficient algorithms for variants of classic hard problems
such as 3SUM (Section~\ref{sec:hardness}). These hardness arguments
suggest that we must accept a weaker notion of security if we wish to
support a larger class of queries.

\subsection{Oblivious RAM Simulations} 
\label{sec:oram}
ORAM simulations first proposed by Goldreich and Ostrovsky
\cite{DBLP:journals/jacm/GoldreichO96} is a general technique for
making memory accesses oblivious that works for arbitrary
programs. Specifically, ORAM simulation is the online transformation
of an arbitrary program $P$ to an equivalent program $P'$ whose memory
accesses appear random (more precisely, drawn from some distribution
that depends only on the number of memory accesses of $P$). By running
$P'$ within a \emph{secure CPU (TM)} and using suitable encryption, an
adversary observing the sequence of memory accesses to an untrusted
memory learns nothing about $P$ and its data other than its number of
memory accesses.  Current ORAM simulation techniques work by adding a
virtualization layer that continuously shuffles (untrusted) memory
contents and adds spurious memory accesses, so that the resulting
access pattern looks random.

A natural idea for oblivious query processing, implemented in a recent
system~\cite{phantom}, would be to run a standard query processing
algorithm under ORAM simulation. However, the resulting query
processing is \emph{not} secure for our definition of security since
it reveals more than just the output size. ORAM simulation, since it
is designed for general programs, does not hide the total number of
memory accesses; in the context of standard query processing, this
reveals the size of intermediate results in a query
plan. Understanding the utility of this weaker notion of security in
the context of a database system is an interesting direction of future
work.

For database queries that admit polynomial time algorithms
(which includes queries covered by Theorem~\ref{thm:informal}) we can
design oblivious algorithms based on ORAM simulation: the number of
memory accesses of such an algorithm is bounded by some
polynomial\footnote{This argument does not depend on $p$ being a
  polynomial, any function works.} $p(n,m)$, where $n$ is the input
size, and $m$, the output size. We modify the algorithm with dummy
memory accesses so that the number of memory accesses for any instance
with input size $n$ and output size $m$ is exactly $p(n,m)$. An ORAM
simulation of the modified algorithm is oblivious. We note that we
need to precisely specify $p(n,m)$ upto constants (not
asymptotically), otherwise the number of memory accesses would be
slightly different for different $(n,m)$ intances making the overall
algorithm non-oblivious. In practice, working out a precise
upper-bound $p(n,m)$ for arbitrarily complex queries is a non-trivial
undertaking.

Our algorithms which are designed to exploit the structure and
semantics of queries have significant performance benefits over the
ORAM-based technique sketched above given the current state-of-the-art
in ORAM simulation. For simplicity, assume for this discussion that
the query output size $m = O(n)$. For small TM memory
($\mathrm{polylog}(n)$), the current best ORAM simulation
techniques~\cite{pathoram, DBLP:conf/soda/KushilevitzLO12} incur an
overhead of $\Theta (\log^2 n)$ memory accesses per memory access of
the original algorithm. This implies that the time complexity of any
ORAM-based query processing algorithm is lower-bounded by $\Omega (n
\log^2 n)$.  In contrast, our algorithms have a time complexity of
$O (n \log n)$ and use $O(\log n)$ TM memory.

Also, by construction ORAM simulation randomly sprays memory accesses
and destroys locality of reference, reducing effectiveness of caching
and prefetching in a memory hierarchy. In a disk setting, a majority
of memory accesses result in a random disk seek and we can show that
any ORAM-based query processing algorithm incurs $\Omega ( \frac{n}{B
  \log M} \log^2 \frac{n}{B})$ disk seeks, where $M$ denotes the size
of TM memory. In contrast, all of our algorithms are scan-based except
for the oblivious sorting, which incurs $O (\log_{M/B}
(n/B)) \cdot o(n/B)$ seeks.


\section{Problem formulation}
\label{sec:prelim}

\subsection{Database Preliminaries}
\label{sec:dbprelim}
A \emph{relation schema}, $R(\bar{A})$, consists of a relation symbol
$R$ and associated attributes $\bar{A} = (A_{1}, \ldots, A_{k})$; we
use $\attrs(R)$ to denote the set of attributes $\{A_1, \ldots, A_k\}$
of $R$. An attribute $A_i$ has an associated set of values called its
\emph{domain}, denoted $\domain{A_i}$. We use $\domain{R}$ to denote
$\domain{\bar{A}} = \domain{A_1} \times \ldots \times \domain{A_k}$. A
\emph{database schema} is a set of relation schemas $R_1, \ldots,
R_m$. A (relation) instance corresponding to schema $R(A_1, \ldots,
A_k)$ is a bag (multiset) of \emph{tuples} of the form $\langle a_1,
\ldots, a_k \rangle$ where each $a_i \in \domain{A_i}$. A database
instance is a set of relation instances. In the following we abuse
notation and use the term relation (resp. database) to denote both
relation schema and instance (resp. database schema and instance). We
sometimes refer to relations as \emph{tables} and attributes as
\emph{columns}.

Given a tuple $t \in R$ and an attribute $A \in \attrs(R)$, $t[A]$
denotes the value of the tuple on attribute $A$; as a generalization
of this notation, if $\Acal \subseteq \attrs(R)$ is a set of
attributes, $t[\Acal]$ denotes the tuple $t$ restricted to attributes
in $\Acal$.

We consider two classes of database queries. A
\emph{select-project-join (SPJ)} query is of the form $\project_\Acal
(\select_P (R_1 \join \cdots \join R_q))$, where the projection
$\project$ is duplicate preserving (we use multiset semantics for all
queries) and $\join$ refers to the natural join.  For $R_1 \join
R_2$, each tuple $t_1 \in R_1$ joins with each tuple $t_2 \in R_2$
such that $t_1 [\attrs(R_1) \isect \attrs(R_2)] = t_2 [\attrs(R_1)
\isect \attrs(R_2)]$ to produce an output tuple $t$ over attributes
$\attrs(R_1) \union \attrs(R_2)$ that agrees with $t_1$ on attributes
$\attrs(R_1)$ and with $t_2$ on attributes $\attrs(R_2)$. The second
class of queries involves grouping and aggregation and is of the form
$\groupagg_\Gcal^{F(A)} (\select_P (R_1 \join \cdots \join R_q))$ and
we call such queries GSPJ queries.  Given relation $R$, $\Gcal
\subseteq \attrs(R)$, $A \in \attrs(R)$, $\groupagg_\Gcal^{F(A)} (R)$,
represents grouping by attributes in $\Gcal$ and computing aggregation
function $F$ over attribute $A$.

\subsection{Secure Query Processing}
\label{sec:sqp-def}
A \emph{relation encryption scheme} is used to encrypt relations. It
is a triple of polynomial algorithms $(\textsf{Enc}, \textsf{Dec},
\textsf{Gen})$ where $\textsf{Gen}$ takes a security parameter $k$ and
returns a key $K$; $\textsf{Enc}$ takes a key $K$, a plaintext
relation instance $R$ and returns a ciphertext relation $\Ccal_R$;
$\textsf{Dec}$ takes a ciphertext relation $\Ccal_R$ and key $K$ and
returns plaintext relation $R$ if $K$ was the key under which
$\Ccal_R$ was produced. A relation encryption scheme is also a
database encryption scheme: to encrypt a database instance we simply
encrypt each relation in the database.  

Informally, a relation encryption scheme is \emph{IND-CPA secure} if a
polynomial time adversary with access to encryption oracle cannot
distinguish between the encryption of two instances $R^{(1)}$ and
$R^{(2)}$ of relation schema $R$ such that $|R^{(1)}| = |R^{(2)}|$
($|R^{(1)}|$ denotes the number of tuples in $R^{(1)}$). Assuming all
tuples of a given schema have the same representational length (or can
be made so using padding), we can construct IND-CPA secure relation
encryption by encrypting each tuple using a standard encryption scheme
such as AES in CBC mode (which is believed to be IND-CPA secure for
message encryption). The detail that encryption is at a
tuple-granularity is relevant for our algorithms which assume that we
can read and decrypt one tuple at a time.

Our formal definition of secure query processing captures: (1)
Database security: An adversary with knowledge of a query does not
learn anything other than the result size of the query by observing
query execution; (2) Query security: An adversary without knowledge of
the query does not learn the constants in the query from query
execution. Appendix~\ref{sec:query-sec} contains a discussion of query
security.

A \emph{query template} $\Qcal$ is a set of queries that differ only
in constants. An example template is the set $\{ \select_{A=1} (R),$
$\select_{A=2} (R), \cdots \}$ which we denote $\select_{A=*} (R)$.

A query processing algorithm $\mathbb{A}_\Qcal$ for a query template
$\Qcal$ takes as input an encrypted database instance $\textsf{Enc}_K
(D)$, a query $Q \in \Qcal$ and produces as output\linebreak $\textsf{Enc}_K (Q
(D))$; Algorithm $\mathbb{A}_\Qcal$ has access to encryption key $K$
and the encryption scheme is IND-CPA secure. Our goal is to make
algorithm $\mathbb{A}_\Qcal$ secure against a passive adversary who
observes its execution. Algorithm $\mathbb{A}_\Qcal$ runs within the
trusted module TM. The TM also has a small amount of internal storage
invisible to the adversary. Algorithm $\mathbb{A}_\Qcal$ has access to
a large amount of untrusted storage which is sufficient to store
$\textsf{Enc}_K (D)$ and any intermediate state required by
$\mathbb{A}_\Qcal$.  The \emph{trace} of an execution of algorithm
$\mathbb{A}_\Qcal$ is the sequence of untrusted memory accesses
$\mathit{read} (i)$ and $\mathit{write} (i, \mathit{value})$, where
$i$ denotes the memory location.

We define security of algorithm $\mathbb{A}_\Qcal$ using the following
\emph{indistinguishability} experiment:
\begin{enumerate}
\item Pick $K \leftarrow \textsf{Gen} (1^k)$
\item The adversary $\Acal$ picks two queries $Q_1 \in \Qcal$, $Q_2
  \in \Qcal$ with the same template and two database instances
  $D^{(1)} = \{R_1^{(1)}, \ldots,$ $R_n^{(1)} \}$ and $D^{(2)} =$\linebreak
  $\{R_1^{(2)}, \ldots, R_n^{(2)}\}$ having the same schema such that
  (1) $|R_i^{(1)}| = |R_i^{(2)}|$ for all $i \in [1, n]$; and (2)\linebreak
  $|Q_1 (D^{(1)})| = |Q_2 (D^{(2)})|$.
\item Pick a random bit $b \leftarrow \{0,1\}$ and let $\tau_b$
  denote the trace of $\mathbb{A}_\Qcal (\textsf{Enc}_K (D^{(b)}),$
  $Q_b)$.
\item The adversary $\Acal$ outputs prediction $b'$ given $\tau_b$, 
  $\textsf{Enc}_K (D^{(b)})$, and $\textsf{Enc}_K (Q_b (D^{(b)}))$.
\end{enumerate}
We say adversary $\Acal$ succeeds if $b' = b$. Algorithm
$\mathbb{A}_\Qcal$ is secure if for any polynomial time adversary
$\Acal$, the probability of success is at most $1/2 + \mathit{negl}
(k)$ for some negligible function\footnote{A negligible function
  $\textit{negl}(k)$ is a function that grows slower than
  $\frac{1}{p(k)}$ for any polynomial $p(k)$; e.g., $\frac{1}{2^k}$.}
$\mathit{negl}$.  We note that our definition of security captures
both database security, since an adversary can pick $Q_1 = Q_2$, and
query security, since an adversary can pick $D^{(1)} = D^{(2)}$.

\subsection{Oblivious Query Processing}
As discussed in Section~\ref{sec:intro}, oblivious query processing is
a simpler notion that is equivalent to secure query processing. Fix an
algorithm $\mathbb{A}_\Qcal$. For an input $I = \textsf{Enc}_K (D)$
and query $Q \in \Qcal$, the memory access sequence
$\Mcal_{\mathbb{A}_\Qcal} (I, Q)$ is the sequence of UM memory reads
$r(i)$ and writes $w(i)$, where $i$ denotes the memory location; the
value being read/written is not part of\linebreak $\Mcal_{\mathbb{A}_\Qcal} (I,
Q)$. In general, $\mathbb{A}_\Qcal$ is randomized and\linebreak
$\Mcal_{\mathbb{A}_\Qcal} (I, Q)$ is a random variable defined over
all possible memory access sequences. Algorithm $\mathbb{A}_\Qcal$ is
oblivious if the distribution of its memory access sequences is
independent of database contents once we fix the query output and
database size.  Formally, Algorithm $\mathbb{A}_\Qcal$ is oblivious if
for any memory access sequence $M$, any two queries $Q_1, Q_2 \in
\Qcal$, any two database encryptions $I_1 = \textsf{Enc}_{K_1}
(D^{(1)})$, $I_2 = \textsf{Enc}_{K_2} (D^{(2)})$:
\[ \mathrm{Pr} [\Mcal_{\mathbb{A}_\Qcal} (I_1, Q_1) = M] = \mathrm{Pr}
[\Mcal_{\mathbb{A}_\Qcal} (I_2, Q_2) = M] \] where $D^{(1)} =
\{R_1^{(1)}, \ldots, R_n^{(1)} \}$ and $D^{(2)} = \{R_1^{(2)}, \ldots,$
$R_n^{(2)}\}$ have the same schema and: (1) $|R_i^{(1)}| = |R_i^{(2)}|$
for all $i \in [1, n]$; and (2) $|Q_1 (D^{(1)})| = |Q_2
(D^{(2)})|$. 

Our definition of obliviousness is more stringent than the one used in
ORAM simulation. In ORAM simulation, the memory access distribution
can depend on the total number of memory accesses, while our
definition precludes dependence on the total number of memory accesses
once the query input and output sizes are fixed.  The following
theorem establishes the connection between oblivious and secure query
processing.

\begin{theorem}
\label{thm:obl}
Assuming one-way functions exist, the existence of an oblivious
algorithm for a query template $\Qcal$ implies the existence of a
secure algorithm for $\Qcal$ with the same asymptotic performance
characteristics (TM memory required, running time).
\end{theorem}
The idea of using obliviousness to derive security from access pattern
leakage was originally proposed
in~\cite{DBLP:journals/jacm/GoldreichO96} and the proof of
Theorem~\ref{thm:obl} is similar to the proof of analogous Theorem
3.1.1 in \cite{DBLP:journals/jacm/GoldreichO96}. Informally, we get
secure query processing by ensuring both data security of values
stored in untrusted memory and access pattern obliviousness. Data
security can be achieved by using encryption, and secure encryption
schemes exist assuming the existence of one-way functions. It follows
that the existence of oblivious query processing algorithms implies
the existence of secure algorithms. Based on Theorem~\ref{thm:obl},
the rest of the paper focuses on oblivious query processing and does
not directly deal with encryption and data security.


\section{Intuition}
\label{sec:overview}

\begin{figure}[t]
\centering
\includegraphics[trim=10ex 60ex 20ex 45ex, width=3in]{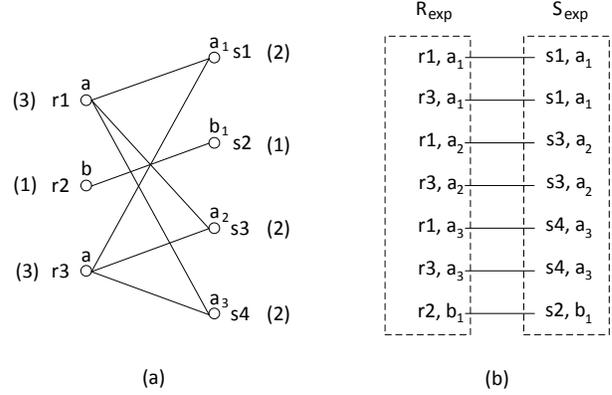}
\caption{Illustration of Oblivious Binary Join}
\label{fig:binjoin}
\end{figure}

This section presents a high level intuition behind our algorithms.
Consider the binary join $R (A, \ldots) \join S (A, \ldots)$ and the
join graph instance shown in Figure~\ref{fig:binjoin}(a). Lower case
letters $a$, $b$, represent values of the joining column $A$; ignore
the subscripts on $a$ and $b$ for now. We add identifiers $r_1$-$r_3$
and $s_1$-$s_4$ to tuples so that we can refer to them in text.

Our oblivious binary join algorithm works in two\linebreak stages: In the first
stage, we compute the contribution of each tuple to the final
output. This is simply the degree of the tuple in the join graph; this
value is shown within parenthesis in Figure~\ref{fig:binjoin}(a). For
example, the degree of $r_1$ is $3$, and degree of $r_2$, $1$. In the
second stage, we \emph{expand} $R$ to $R_{exp}$ by duplicating each
tuple as many times as its degree; $r_1$ occurs $3$ times in
$R_{exp}$, $r_2$ once, and so on. We similarly, expand $S$ to
$S_{exp}$. The expansions $R_{exp}$ and $S_{exp}$ are shown within
boxed rectangles in Figure~\ref{fig:binjoin}(b). The final join output
is produced by ``stitching'' together $R_{exp}$ and $S_{exp}$ as
illustrated in Figure~\ref{fig:binjoin}(b). The expansions $R_{exp}$
and $S_{exp}$ are \emph{sequences} whose ordering is picked carefully
to ensure that stitching the $i$th tuple in $R_{exp}$ with the $i$th
tuple in $S_{exp}$ indeed produces the correct join output.

A central component of the above algorithm are oblivious
implementations of two simple primitives that we call \emph{semi-join
  aggregation} and \emph{expansion}. Semi-join aggregation computes
the degree of each tuple in a join and expansion expands a relation by
duplicating each tuple a certain number of times such as its degree.

The same approach generalizes to multiway joins if the overall query
is \emph{acyclic}~\cite{DBLP:conf/pods/PaghP06}. Informally, to
compute $R \join S \join T$, we would compute the contribution of each
tuple to the \emph{final} join output and use these values to expand
input relations to $R_{exp}$, $S_{exp}$, and $T_{exp}$, which are then
stitched together to produce the final join output.


\section{Primitives}
\label{sec:prim}
This section introduces a few core primitives and\linebreak presents oblivious
algorithms---algorithms that have the same UM memory access patterns
once we fix input and output sizes---for these primitives. These
primitives serve two purposes: First, as discussed in
Section~\ref{sec:overview}, they are building blocks for our oblivious
query processing algorithms; Second, they introduce notation to help
us concisely specify our algorithms, and reason about their
obliviousness and performance.

There exist oblivious algorithms for all the primitives of this
section having time complexity $O (n \log n)$ and requiring $O(\log
n)$ TM memory, where $n$ denotes the sum of input and output
sizes. Some of these algorithms rely on an oblivious sort; an optimal
$O (n \log n)$ oblivious sort algorithm that uses $O(1)$ TM memory is
presented in~\cite{DBLP:journals/jacm/Goodrich11}. Due to space
constraints we defer presenting oblivious algorithms for the simpler
primitives to the full version of the paper~\cite{full-version}.

\begin{figure}
\small
\centering
\begin{tabular}{ccc}
\begin{tabular}{c|c|} \cline{2-2}
& $A$ \\ \cline{2-2}
$r_1$ & $1$ \\ \cline{2-2}
$r_2$ & $2$ \\ \cline{2-2}
$r_3$ & $1$ \\ \cline{2-2}
\end{tabular}
    &     
\begin{tabular}{c|c|c|} \cline{2-3}
      & $A$ & $B$ \\ \cline{2-3} 
$s_1$ & $1$ & $1$ \\ \cline{2-3}
$s_2$ & $2$ & $1$ \\ \cline{2-3}
$s_3$ & $2$ & $1$ \\ \cline{2-3}
\end{tabular}

    & 
\begin{tabular}{|c|c|c|c|c|} \hline
$A$ & $B$ & $C$ & $D$ & $E$ \\ \hline
$1$ & $2$ & $1$ & $2$ & $1$ \\ \hline
$2$ & $4$ & $1$ & $6$ & $2$ \\ \hline
$1$ & $2$ & $2$ & $8$ & $1$ \\ \hline
\end{tabular}

\\

$R$ & $S$ & $\Rtil$
\end{tabular}
\caption{Illustration of primitives: $\Rtil \gets R.(B \gets 2A).(C
  \gets \idfunc_A).(D \gets \rsfunc (B)).(E \stackrel{\ltimes}{\gets} \mathit{Sum}(S.B))$. $r_1$-$r_3$ and
$s_1$-$s_3$ are names we use to refer to the tuples.}
\label{fig:primex}
\end{figure}

\vspace{1ex}\noindent\textbf{Relation Augmentation: }
This primitive adds a derived column to a relation. In the simplest
form the derived column is obtained by applying a function to existing
columns; many primitives we introduce subsequently are more complex
instantiations of relation augmentation. We use the notation $R.(A
\gets \mathit{Func})$ to represent relation augmentation which adds a
new derived column $A$ using some function $\mathit{Func}$ and
produces an output relation with schema $\attrs(R) \union \{A\}$. For
example, $R.(B \gets 2A)$ adds a new column $B$ whose value is twice
that of $A$ (see Figure~\ref{fig:primex}). Our notation allows
composition to be expressed more concisely; e.g., $R.(B \gets 2A).(F
\gets A + B)$. 

\vspace{1ex}\noindent\textbf{Grouping Identity: } This relation
augmentation primitive adds a new \emph{identity} column within a
group; identity column values are of the form $1, 2, \ldots$.  In
particular, we use the notation $R.(A \gets \idfunc_{\Gcal}^{\Ocal})$
where $\Gcal \subseteq \attrs(R)$ is a set of grouping columns and
$\Ocal \subseteq \attrs(R)$ is a set of ordering columns. To get the
output, we partition the tuples by the grouping columns $\Gcal$, order
the tuples within each partition by $\Ocal$, and assign ids based on
this ordering. (We break ties arbitrarily, so the output can be
non-deterministic.) $\Gcal$ and $\Ocal$ can be empty and omitted. For
example, $R.(Id \gets \idfunc)$ assigns an unique id to each record in
$R$. In Figure~\ref{fig:primex}, for $R.(C \gets \idfunc_A)$, we
partition by $A$, so $r_1$ and $r_3$ go to the same partition; tuple
$r_1$ gets a $C$ value of $1$, and $r_3$, a $C$ value of $2$.

\vspace{1ex}\noindent\textbf{Grouping Running Sum: }
This primitive is a generalization of grouping identity and adds a
running sum column to a relation. It is represented\linebreak $R.(A \gets
\rsfunc_{\Gcal}^{\Ocal} (B))$; it groups a relation by $\Gcal$ and
orders tuples in a group by $\Ocal$ and stores the running sum of $B$
column values in a new column $A$. In particular, grouping identity
$R.(\mathit{Id} \gets \idfunc_{\Gcal}^{\Ocal})$ can be expressed as
$R.(X \gets 1).(\mathit{Id} \gets \rsfunc_\Gcal^\Ocal (X))$. See
Figure~\ref{fig:primex} for an example.

\vspace{1ex}\noindent\textbf{Generalized Union: } A \emph{generalized
  union} of $R$ and $S$, denoted $R \gunion S$, produces a relation
with schema $\attrs(R) \union \attrs(S)$ that contains tuples from
both $R$ and $S$. Tuples of $R$ have a \emph{null} value for
attributes in $\attrs(S) - \attrs(R)$, and those of $S$, a null value
for attributes in $\attrs(R) - \attrs(S)$.

\vspace{1ex}\noindent\textbf{Sequences: Sorting and Stitching: }
Although the inputs and outputs of our algorithms are relations
represented as \emph{sequences}, the ordering is often unimportant and
we mostly do not emphasize the sequentiality. We use the notation
$\langle R \rangle$ to represent some sequence corresponding to
$R$. When a particular ordering is desired, we represent the ordering
as $\langle R \rangle_\Ocal$ where $\Ocal \subseteq \attrs (R)$ denote
the ordering attributes.

One operation on sequences that cannot be represented over bags is
``stitching'' two sequences of the same length (see
Figure~\ref{fig:binjoin}(b) for an example): Given two sequences
$\langle R \rangle$ and $\langle S \rangle$ of the same length $n$,
the operation $\langle R \rangle \cdot \langle S \rangle$ produces a
sequence of length $n$ with schema $\attrs(R) \union \attrs(S)$ and
the $i$th tuple of the sequence is obtained by concatenating the $i$th
tuple of $\langle R \rangle$ and the $i$th tuple of $\langle S
\rangle$; we ensure when invoking this operation that the $i$th tuples
of both sequences agree on $\attrs(R) \isect \attrs(S)$ if the
intersection is nonempty.

\vspace{1ex}\noindent\textbf{Filters:} Consider the filter $\select_P
(R)$. The simple algorithm that scans each tuple $t \in R$, checks if
it satisfies $P$, and outputs it if does, is \emph{not}
oblivious. (E.g., simply reordering tuples in $R$ changes the memory
write pattern.)

The oblivious sorting algorithm can be used to design a simple
oblivious algorithm for selection (filter). To evaluate $\select_P
(R)$, we sort $R$ such that tuples that satisfy predicate $P$ occur
before tuples that do not. We scan the sorted table and output the
tuples that satisfy $P$ and stop when we encounter the first tuple
that does not satisfy $P$. The overall data access pattern depends
only on input and output sizes and is therefore oblivious.

\subsection{Semi-Join Aggregation}
\label{sec:semi-join-aggr}
Semi-join aggregation, denoted $R.(A \stackrel{\ltimes}{\gets}
\mathit{Sum}(S.B))$, is equivalent\footnote{This equality holds only
  when $R$ has not duplicates.} to the relational algebra expression\linebreak
$\groupagg_{\attrs(R)}^{A \gets \mathrm{SUM}(S.B)} (R \join S)$. This
operation adds a new derived column $A$; for each tuple $t_R \in R$,
we obtain value of $A$ by identifying all $t_S \in S$ that join with
$t_R$ (agree on all common attributes $\attrs (R) \isect \attrs (S)$)
and summing over $t_S [B]$ values. As discussed in
Section~\ref{sec:overview}, we introduce this primitive to compute the
degree of a tuple in a join graph. In particular, the degree of each
$R$ tuple in $R \join S$ is obtained by $\Stil \gets S.(X \gets 1)$,
$R.(\mathit{Degree} \stackrel{\ltimes}{\gets}
\mathit{Sum}(\Stil.X))$. In Figure~\ref{fig:primex}, $r_2$ joins with
two tuples $s_2$ and $s_3$, so $r_2[E]$ is $s_2[B] + s_3[B] = 2$.

\begin{algorithm}[t]
  \small
  \caption{Semi-Join Aggregation: $R.(X \stackrel{\ltimes}{\gets} \mathit{Sum}(S.Y))$}
  \label{alg:semi-join-agg}
  \begin{algorithmic}[1]
    \Procedure{SemiJoinAgg}{$R, S, X, Y$}
      \State $\Rtil \gets R.(\mathit{Src} \gets 1).(Y \gets 0)$ \label{stp:lineageR}
      \State $\Stil \gets S.(\mathit{Src} \gets 0)$ \label{stp:lineageS}
      \State $U \gets \Rtil\ \gunion\ \Stil$ \label{stp:sjgenunion}
      \State $U \gets U.(X \gets \rsfunc_{\attrs (R) \isect
        \attrs(S)}^{\mathit{Src}} (Y))$  \label{stp:sja-rsum}
      \State Output $\project_{\attrs (R), X} (\select_{\mathit{Src} =
        1} (U))$\label{stp:disunion}
    \EndProcedure
  \end{algorithmic}
\end{algorithm}

\vspace{1ex}\noindent\textbf{Oblivious Algorithm:}
Algorithm~\ref{alg:semi-join-agg} presents an oblivious algorithm for
semi-join aggregation\linebreak $R.(X \stackrel{\ltimes}{\gets}
\mathit{Sum}(S.Y))$. (In all our algorithms, each step involves one of
our primitives and is implemented using the oblivious algorithm for
the primitive.) It adds a ``lineage'' column $\mathit{Src}$ in Steps
\ref{stp:lineageR} and \ref{stp:lineageS}; the value of $\mathit{Src}$
is set to $1$ for all $R$ tuples and $0$ for all $S$ tuples. A $Y$
column initialized to $0$ is added to all $R$ tuples. Step
\ref{stp:sjgenunion} computes a generalized union $U$ of $\Rtil$ and
$\Stil$. Adding the running sum within each $\attrs(R) \isect
\attrs(S)$ group adds the required aggregation value into each $R$
tuple (Step~\ref{stp:sja-rsum}); the running sum computation is
ordered by $\mathit{Src}$ to ensure that all $S$ tuples within an
$\attrs(R) \isect \attrs(S)$ group occur before the $R$
tuples. Finally, the oblivious filter $\select_{\mathit{Src} = 1}$ in
Step~\ref{stp:disunion} extracts the $R$ tuples from
$U$. Figure~\ref{fig:exsja} shows the intermediate tables generated by
Algorithm~\ref{alg:semi-join-agg} for sample tables $R (Id, A)$ and $S
(A, Y)$.

\begin{figure}
\small
\centering
\begin{tabular}{ccc}

\hspace{-3ex}
\begin{tabular}{|c|c|c|c|} \hline
$\mathit{Id}$ & $A$ & $Y$ & $\mathit{Src}$ \\ \hline
1 & a & 0 & 1 \\ \hline
2 & b & 0 & 1 \\ \hline
\end{tabular}

&

\hspace{-3ex}
\begin{tabular}{|c|c|c|} \hline
 $A$ & $Y$ & $\mathit{Src}$ \\ \hline
 a & 2 & 0 \\ \hline
 b & 3 & 0 \\ \hline
 a & 4 & 0 \\ \hline
\end{tabular}

&

\hspace{-3ex}
\begin{tabular}{|c|c|c|c|c|} \hline
$\mathit{Id}$ & $A$ & $Y$ & $\mathit{Src}$ & $X$ \\ \hline
- & a & 2 & 0 & 2 \\ \hline
- & a & 4 & 0 & 6 \\ \hline
1 & a & 0 & 1 & 6 \\ \hline
- & b & 3 & 0 & 3 \\ \hline
2 & b & 0 & 1 & 3 \\ \hline
\end{tabular}

\\

(a): $\Rtil$ & (b): $\Stil$ & (c): $U$
\end{tabular}
\caption{Sample computation of $R.(X \stackrel{\ltimes}{\gets} \mathit{Sum}(S.Y))$}
\label{fig:exsja}
\end{figure}

\begin{theorem}
\label{thm:ob-sja}
Algorithm~\ref{alg:semi-join-agg} obliviously computes semi-join
aggregation $R.(X \stackrel{\ltimes}{\gets} \mathit{Sum}(S.Y))$ of two
tables in $O ((n_R + n_S) \log (n_R + n_S))$ time and using
$O(1)$ TM memory, where $n_R = |R|$ and $n_S = |S|$ denote the input
table sizes.
\end{theorem}
\begin{proof} (Sketch) For each step of
  Algorithm~\ref{alg:semi-join-agg} the input and output sizes are one
  of $n_R$, $n_S$, and $n_R + n_S$ and each step is locally oblivious
  in its input and output sizes. The overall algorithm is therefore
  oblivious. Further, the oblivious algorithms for each step require
  $O(1)$ TM memory.
\end{proof}

\subsection{Expansion}
\label{sec:expansion-1}
This primitive duplicates each tuple of a relation a number of times
as specified in one of the columns. In particular, the output of
$\expand_W (R)$, $W \in \attrs(R)$ and $\domain{W} \subseteq
\mathbb{N}$ is a relation instance with same schema, $\attrs(R)$, that
has $t [W]$ copies of each tuple $t \in R$. For example, given an
instance of $R(A,W): \{ \langle a, 1 \rangle, \langle b, 2 \rangle
\}$, $\expand_W (R)$ is given by $\{ \langle a, 1 \rangle, \langle b,
2 \rangle, \langle b, 2 \rangle \}$. As discussed in
Section~\ref{sec:overview}, expansion plays a central role in our join
algorithms.

\subsubsection{Oblivious Algorithm }
\label{sec:expansion}
We now present an oblivious algorithm to compute $\expand_W
(R)$. For presentational simplicity, we slightly modify the
representation of the input. The modified input to the expansion is a
sequence of pairs $(\langle r_1, w_1 \rangle,$ $\ldots, \langle r_n, w_n
\rangle)$, where $r_i$s are values (tuples) drawn from some domain and
$w_i \in \mathbb{N}$ are non-negative \emph{weights}. The desired
output is some sequence containing (in any order) $w_i$ copies of each
$r_i$. We call such a sequence a \emph{weighted sequence}.

The input size of expansion is $n$ and the output size is $m
\stackrel{\mathrm{def}}{=} \sum_{i=1}^n w_i$, so memory access pattern
of an oblivious algorithm depends on only these two quantities. The
naive algorithm that reads each $\langle r_i, w_i \rangle$ into TM and
writes out $w_i$ copies of $r_i$ is not oblivious, since the output
pattern depends on individual weights $w_i$.

We first present an oblivious algorithm when the input sequence has a
particular property we call \emph{prefix-heavy}; we use this algorithm
as a subroutine in the algorithm for the general case.
\begin{definition}
  \label{def:prefix-heavy}
  A weighted sequence $(\langle r_1, w_1 \rangle,$ $\ldots,$\linebreak
  $\langle r_n, w_n \rangle)$ is \emph{prefix-heavy} if for each $\ell
  \in [1,n]$,\linebreak $\frac{1}{\ell} \sum_{i = 1}^\ell w_i \geq \frac{1}{n}
  \sum_{i = 1}^n w_i$.
\end{definition}
The average weight of any prefix of a prefix-heavy sequence is
greater-than-or-equal to the overall average weight. Any weighted
sequence can be reordered to get a prefix-heavy sequence, e.g., by
sorting by non-decreasing weight. The sequence $( \langle a, 4
\rangle, \langle b, 1 \rangle,$ $\langle c, 2 \rangle )$ is
prefix-heavy, while $( \langle b, 1 \rangle,$ $\langle a, 4 \rangle,
\langle c, 2 \rangle )$ is not.
\begin{algorithm}[t]
  \small
  \caption{Oblivious Expansion of prefix heavy sequences}
  \label{alg:obl-exp-ph}
  \begin{algorithmic}[1]
    \Procedure{ExpandPrefixHeavy}{$I$}
      \Statex \textbf{Assume:} $I = ( \langle r_1, w_1 \rangle, \ldots,
    \langle r_{n}, w_{n} \rangle )$
      \Statex \textbf{Require:} $I$ is prefix heavy
      \State $w_{\mathit{avg}} \gets (\sum_{i=1}^{n} w_{i}) / n$ 
      \State $\Ccal_{TM} \gets \phi$ \Comment{counters within TM}
      \ForIter{$i = 1$}{$n$}
        \State Read $\langle r_i, w_i \rangle$ to TM. \label{stp:input}
        \State $w_{\mathit{curr}} \gets \lfloor i \cdot w_{\mathit{avg}}
        \rfloor - \lfloor (i - 1)  \cdot w_{\mathit{avg}} \rfloor$ \label{stp:wcurr}
        \If{$w_i \leq w_{\mathit{curr}}$} 
          \State Append $w_i$ copies of $r_i$ to output
          \State $w_{\mathit{curr}} \gets w_{\mathit{curr}} - w_i$
        \Else 
          \State $\Ccal_{TM} [r_i] \gets w_i$
        \EndIf
        \While{$w_{\mathit{curr}} > 0$}
          \State $r_j \gets \mathrm{argmin}_k$ $r_k$ has a counter in
          $\Ccal_{TM}$ 
          \If{$\Ccal_{TM} [r_j] > w_{\mathit{curr}}$}
            \State Append $w_{\mathit{curr}}$ copies of $r_j$ to output
            \State $\Ccal_{TM} [r_j] \gets \Ccal_{TM} [r_j] -
            w_{\mathit{curr}}$
            \State $w_{\mathit{curr}} \gets 0$
          \Else \label{stp:zerobegin}
            \State Append $\Ccal_{TM} [r_j]$ copies of $r_j$ to
            output
            \State $w_{\mathit{curr}} \gets w_{\mathit{curr}} -
            \Ccal_{TM} [r_j]$
            \State // $\Ccal_{TM} [r_j] \gets 0$
            \State Remove $r_j$ from $\Ccal_{TM}$
          \EndIf \label{stp:zeroend}
        \EndWhile
      \EndForIter
    \EndProcedure
  \end{algorithmic}
\end{algorithm}

Algorithm~\ref{alg:obl-exp-ph} presents an oblivious algorithm for
expanding prefix heavy weighted sequences.  To expand the sequence $I
= ( \langle r_1, w_1 \rangle,$ $\ldots, \langle r_{n}, w_{n} \rangle
)$, the algorithm proceeds in $n$ (input-output) steps. Let
$w_{\mathit{avg}} = (\sum_{i=1}^n w_i)/n$ denote the average weight of
the sequence.  In each step, it reads one weighted record
(Step~\ref{stp:input}) and produces $w_{\mathit{avg}}$ (unweighted)
records in the output; the actual number $w_{\mathit{curr}}$ is either
$\lfloor w_{\mathit{avg}} \rfloor$ or $\lceil w_{\mathit{avg}}
\rceil$, when $w_{\mathit{avg}}$ is fractional (Step~\ref{stp:wcurr}).

Call a record $\langle r_i, w_i \rangle$ \emph{light} if $w_i \leq
w_{\mathit{avg}}$ and \emph{heavy}, otherwise. If the current record
$\langle r_i, w_i \rangle$ is light, $w_i$ copies of $r_i$ are
produced in the output; if it is heavy, a counter $\Ccal[r_i]$ is
initialized with count $w_i$ denoting the number of copies of $r_i$
available for (future) outputs. Previously seen heavy records are used
to make up the ``balance'' and ensure $w_{\mathit{avg}}$ records are
produced in each step. The counters $\Ccal$ are internal to TM and are
not part of the data access pattern. Figure~\ref{fig:obl-exp-ph} shows
the steps of Algorithm~\ref{alg:obl-exp-ph} for the sequence $(
\langle a, 4 \rangle, \langle b, 1 \rangle, \langle c, 2 \rangle )$.

\begin{figure}
\small
\centering
\begin{tabular}{|c|c|c|c|} \hline
Step & Input & Output & Counters \\ \hline
1 & $\langle a, 4 \rangle$ & $a, a$ & $\Ccal[a] = 2$ \\ \hline
2 & $\langle b, 1 \rangle$ & $b, a, a$ & $\Ccal = \phi$ \\ \hline
3 & $\langle c, 2 \rangle$ & $c, c$ & $\Ccal = \phi$ \\ \hline
\end{tabular}
\caption{Algorithm~\ref{alg:obl-exp-ph} over sequence $( \langle a, 4
\rangle, \langle b, 1 \rangle, \langle c, 2 \rangle )$}
\label{fig:obl-exp-ph}
\end{figure}

Algorithm~\ref{alg:obl-exp-ph} is oblivious since its input-output
pattern is fixed once the input size $n$ and output size $m = \sum_{i
  =1}^n w_i$ is fixed. Note that $w_{\mathit{avg}} = m/n$ is fixed
once input and output sizes are fixed. 

In the worst case, the number of counters maintained by
Algorithm~\ref{alg:obl-exp-ph} can be $\Omega (n)$.
\begin{example}
\label{ex:algph-worst}
Consider the sequence $w_1 = \cdots = w_{n/4} = 4$ and $w_{n/2+1} =
\cdots = w_n = 0$. After reading $n/4$ records, we can show that
Algorithm~\ref{alg:obl-exp-ph} requires $\approx 3n/16$ counters.
\end{example}
However, any weighted sequence can be re-ordered so that it is prefix
heavy and the number of counters used by
Algorithm~\ref{alg:obl-exp-ph} is $O(1)$ as stated in
Lemma~\ref{lem:space-order} and illustrated in the following example.
\begin{example}
\label{ex:bph}
We can reorder the weight sequence in Example~\ref{ex:algph-worst} as
$\langle 4, 0, 0, 0, 4, 0, 0, 0, \ldots, \rangle$ interleaving $3$
light records inbetween two heavy records. We can show that
Algorithm~\ref{alg:obl-exp-ph} requires just one counter for this
sequence.
\end{example}
More generally, the basic idea is to interleave sufficient number of
light records between two heavy records so that average weight of any
prefix is barely above the overall average, which translates to fewer
number of counters. In Example~\ref{ex:bph}, we can suppress just one
heavy record to make the average weight of any prefix $\leq
w_{\mathit{avg}}$. We call such sequences \emph{barely prefix heavy}.

\begin{lemma}
\label{lem:space-order}
Any weighted sequence $I$ can be re-ordered as a prefix heavy sequence
$I'$ such that Algorithm~\ref{alg:obl-exp-ph} requires $O(1)$ counters
to process $I'$.
\end{lemma}

Lemma~\ref{lem:space-order} suggests that we can design a general
algorithm for expansion by first reordering the input sequence to be
barely prefix heavy and using Algorithm~\ref{alg:obl-exp-ph}. The main
difficulty lies in \emph{obliviously} reordering the sequence to make
it barely prefix heavy. We do not know how to do this directly;
instead, we transform the input sequence to a modified sequence by
\emph{rounding} weights (upwards) to be a power of $2$. We can
concisely represent the full rounded weight distribution using
logarithmic space. We store this rounded distribution within TM and
use it to generate a barely prefix heavy sequence.  Details of the
algorithm to reorder a sequence to make it barely prefix heavy is
presented in the full-version of the paper~\cite{full-version}.

\begin{algorithm}[t]
  \small
  \caption{Oblivious Expansion $\expand_W (R)$}
  \label{alg:obl-exp}
  \begin{algorithmic}[1]
    \Procedure{Expand}{$R, W$}
      \State $m \gets \groupagg^{\mathrm{SUM}(W)} $ \Comment{output
        size}
      \State $\Rtil \gets R.(\Wtil \gets 2^{\lceil \log_2 W \rceil})$
      \Comment{weight rounding} \label{stp:rounding}
      \State $\mtil \gets \groupagg^{\mathrm{SUM}(\Wtil)}$
      \Comment{assert: $\mtil < 2 m$}
      \State $\Rtil \gets \Rtil\ \gunion\ \{ \langle \mathit{dummy}
      \rangle \}.(W \gets 0).(\Wtil \gets 2m - \mtil)$      \label{stp:dummy}
      \State $\Dcal_{TM} \gets \groupagg_{\Wtil}^{\mathrm{COUNT(*)}}
      (\Rtil)$      \Comment{rounded weight distr} \label{stp:wdist}
      \State $\Rtil \gets \Rtil.(\mathit{Id} \gets \idfunc)$
      \Comment{Attach ids}
      \State $\langle \Rtil_{bph} \rangle \gets
      $\textsc{ReorderBarelyPrefixHeavy}$(\Rtil, \Dtil_{TM})$\label{stp:bph}
      \State $\Rtil_{\mathit{exp}} \gets
      $\textsc{ExpandPrefixHeavy}$(\langle \Rtil_{bph} \rangle)$\label{stp:exp-bph}
      \State $\Rtil_{\mathit{exp}} \gets
      \Rtil_{\mathit{exp}}.(\mathit{Rank} \gets
      \idfunc_{\mathit{Id}})$ \label{stp:filter1}
      \State Output $\project_{\attrs(R)} (\select_{\mathit{Rank} <=
        W} (\Rtil_{\mathit{exp}}))$ \label{stp:filter2}
    \EndProcedure
  \end{algorithmic}
\end{algorithm}
Algorithm~\ref{alg:obl-exp} presents our oblivious expansion
algorithm. Step~\ref{stp:rounding} performs weight rounding. Directly
expanding with these rounded weights produces a sequence of length
$\mtil$; the resulting algorithm would not be oblivious since $\mtil$
does not depend on $n$ and $m$ (output size) alone. We therefore add
(Step~\ref{stp:dummy}) a dummy tuple with rounded weight $2m - \mtil$
(and actual weight $0$). Expanding this table produces $2m$
tuples. Step~\ref{stp:wdist} computes the distribution of rounded
weights. We note that this step consumes $\Rtil$ and produces no
output since $\Dcal_{TM}$ remains within TM. Step~\ref{stp:bph}
reorders the table (sequence) to make it barely prefix heavy which is
expanded using Algorithm~\ref{alg:obl-exp-ph} in
Step~\ref{stp:exp-bph}. Steps \ref{stp:filter1} and \ref{stp:filter2}
filter out dummy tuples produced due to rounding using an oblivious
selection algorithm.
\begin{theorem}
  Algorithm \ref{alg:obl-exp} obliviously expands an input table in
  time $O ((n + m) \log (n + m))$ using $O (\log (n + m))$ TM memory,
  where $n$ and $m$ denote the input and output sizes, respectively.
\end{theorem}
\begin{proof} (Sketch) The input size of each step is one of $n$,
  $n+1$ or $m$. The output size of each step is one of $1$, $n$,
  $n+1$, or $m$. Each step is locally oblivious, so all data access
  patterns are fixed once we fix $n$ and $m$.
\end{proof}


\section{Query Processing Algorithms}
\label{sec:algorithm}
We now present oblivious query processing algorithms for SPJ and GSPJ
queries. Recall from Section~\ref{sec:dbprelim} that these are of the
form $\project_\Acal (\select_P (R_1 \join \cdots \join R_q))$ and
$\groupagg_\Gcal^{F(X)} (\select_P (R_1 \join \cdots \join
R_q))$. Instead of presenting a single algorithm, we present
algorithms for various special cases that together formalize (and
prove) the informal characterization in Theorem~\ref{thm:informal}. We
begin by presenting in Section~\ref{sec:bin-join} an oblivious
algorithm for binary join. In Section~\ref{sec:nary-join-short} we
discuss extensions to multiway joins. Section~\ref{sec:groupagg}
discusses grouping and aggregation, Section~\ref{sec:filter-in-qry}
discusses selection predicates, and Section~\ref{sec:fkey} discusses
how key-foreign key constraints can be exploited.

\subsection{Binary Join}
\label{sec:bin-join}
Recall the discussion from Section~\ref{sec:overview}
(Figure~\ref{fig:binjoin}) which presents the high level intuition
behind our binary join algorithm: Informally, to compute $R (A,
\ldots) \join S (A, \ldots)$ we begin by computing the degree of each
tuple of $R$ and $S$ in the join graph; we use semi-join aggregation
to compute the degree. We then expand $R$ and $S$ to $R_{exp}$ and
$S_{exp}$ by duplicating each tuple of $R$ and $S$ as many times as
its degree. By construction, $R_{exp}$ contains the $R$-half of join
tuples and $S_{exp}$ contains the $S$-half, and we stitch them
together to produce the final join output (see
Figure~\ref{fig:binjoin}(b)).

One remaining detail is to order $R_{exp}$ and $S_{exp}$ so that they
can be stitched to get the join result. Simply ordering by the join
column values does not necessarily produce the correct result. We
attach a subscript to join values on the $S$ side so that different
occurrences of the same value get a different subscript; the three $a$
values now become $a_1, a_2, a_3$. We expand $S$ as before remembering
the subscripts, so there are two copies each of $a_1$, $a_2$, and
$a_3$. We expand $R$ slightly differently: each $a$ tuple on $R$ is
expanded $3$ times and we produce one copy of each subscript. For
example, tuple $r_1$ is expanded to $(r_1, a_1)$, $(r_1, a_2)$ and
$(r_1, a_3)$. Sorting by the subscripted values and stitching produces
the correct join result.

\begin{algorithm}[t]
  \small
  \caption{Binary Natural Join: $R \bowtie S$}
  \label{alg:bin-join}
  \begin{algorithmic}[1]
    \Procedure{BinaryJoin}{$R, S$}
      \State $\Jcal \gets \attrs(R) \isect \attrs(S)$ \Comment{join attrs}
      \State $\Rtil \gets R.(N \gets 1)$ \Comment{tuple multiplicity}
      \State $\Rtil \gets \Rtil.(\mathit{Id} \gets \idfunc)$
      \Comment{Add an id column}
      \State $\Stil \gets S.(N \gets 1)$ \Comment{tuple multiplicity}
      \State $\Stil \gets \Stil.(\mathit{Id} \gets \idfunc)$
      \Comment{Add an id column}
      \State $\Rtil \gets \Rtil.(N_S \stackrel{\ltimes}{\gets}
      \mathit{Sum}(S.N))$\label{stp:rjoindegree} \Comment{Compute degree}
      \State $\Stil \gets \Stil.(N_R \stackrel{\ltimes}{\gets}
      \mathit{Sum}(R.N))$\label{stp:sjoindegree} \Comment{Compute degree}
      \State $\Stil \gets \Stil.(\mathit{JId} \gets \idfunc_\Jcal)$ \label{stp:jid}
      \State $R_{exp} \gets \expand_{N_S} (\Rtil)$ \label{stp:rexp}
      \State $R_{exp} \gets R_{exp}.(\mathit{JId} \gets
      \idfunc_{\mathit{Id}})$ \label{stp:rjid}
      \State $S_{exp} \gets \expand_{N_R} (\Stil)$ \label{stp:sexp}
      \State Output $\project_{\attrs(R) \union \attrs(S)} (
      {\langle R_{exp} \rangle}_{\Jcal, \mathit{JId}}.{\langle S_{exp}
        \rangle}_{\Jcal, \mathit{JId}})$  \label{stp:binjoinstitch}    
    \EndProcedure
  \end{algorithmic}
\end{algorithm}

\vspace{1ex}\noindent{\textbf{Formal Algorithm:}
  Algorithm~\ref{alg:bin-join} presents our join algorithm. Steps
  \ref{stp:rjoindegree} and \ref{stp:sjoindegree} compute the join
  degrees ($N_S$ and $N_R$, resp) for each $R$ and $S$ tuple using a
  semi-join aggregation. Step~\ref{stp:jid} is a grouping identity
  operation. All $S$ tuples agreeing on join columns $\Jcal$ belong to
  the same group, and each gets a different identifier. This step
  plays the role of assigning subscripts to $S$ tuples in
  Figure~\ref{fig:binjoin}(a). Steps \ref{stp:rexp} and \ref{stp:sexp}
  expand $\Rtil$ and $\Stil$ based on the join
  degrees. Step~\ref{stp:rjid} is another grouping identity
  operation. All tuples in $R_{exp}$ that originated from the same $R$
  tuple belong to the same group, and each gets a different
  identifier. This has the effect of expanding $R$ tuples with a
  different subscript. Step~\ref{stp:binjoinstitch} stitches
  expansions of $R$ and $S$ to get the final join output.
  Figure~\ref{fig:binjoinex} illustrates Algorithm~\ref{alg:bin-join}
  for the example shown in Figure~\ref{fig:binjoin}. Note the
  correspondence between $\mathit{Jid}$ column values and subscripts
  in Figure~\ref{fig:binjoin}.
\begin{figure}[t]
\small
\centering
\begin{tabular}{cc}

\begin{tabular}{|c|c|c|c|} \hline
$\mathit{Id}$ & $A$ & $N$ & $N_S$ \\ \hline
$1$ & $a$ & $1$ & $3$ \\ \hline
$2$ & $b$ & $1$ & $1$ \\ \hline 
$3$ & $a$ & $1$ & $3$ \\ \hline
\end{tabular}

& 

\begin{tabular}{|c|c|c|c|c|} \hline
$\mathit{Id}$ & $A$ & $N$ & $N_R$ & $\mathit{JId}$ \\ \hline
$1$ & $a$ & $1$ & $2$ & $1$ \\ \hline
$2$ & $b$ & $1$ & $1$ & $1$ \\ \hline
$3$ & $a$ & $1$ & $2$ & $2$ \\ \hline
$4$ & $a$ & $1$ & $2$ & $3$ \\ \hline
\end{tabular} \\ 

(a): $\Rtil$ & (b): $\Stil$ \\

\begin{tabular}{|c|c|c|} \hline
$\mathit{Id}$ & $A$ & $\mathit{Jid}$ \\ \hline
$1$ & $a$ & $1$  \\ \hline
$3$ & $a$ & $1$  \\ \hline
$1$ & $a$ & $2$  \\ \hline
$3$ & $a$ & $2$ \\ \hline
$1$ & $a$ & $3$ \\ \hline
$3$ & $a$ & $3$ \\ \hline
$2$ & $b$ & $1$ \\ \hline
\end{tabular}

& 

\begin{tabular}{|c|c|c|} \hline
$\mathit{Id}$ & $A$ & $\mathit{Jid}$ \\ \hline
$1$ & $a$ & $1$ \\ \hline
$1$ & $a$ & $1$ \\ \hline
$3$ & $a$ & $2$ \\ \hline
$3$ & $a$ & $2$ \\ \hline
$4$ & $a$ & $3$ \\ \hline
$4$ & $a$ & $3$ \\ \hline
$2$ & $b$ & $1$ \\ \hline
\end{tabular}

\\

(c): ${\langle R_e \rangle}_{A, \mathit{Jid}}$ & (d): ${\langle S_e
  \rangle}_{A, \mathit{Jid}}$
\end{tabular}
\caption{Intermediate tables used by Algorithm~\ref{alg:bin-join} for
  Example of Figure~\ref{fig:binjoin}. Only relevant columns of
  $R_{exp}$ and $S_{exp}$ are shown.}
\label{fig:binjoinex}
\end{figure}

\begin{theorem}
\label{thm:obl-binjoin}
Algorithm~\ref{alg:bin-join} obliviously computes the binary natural
join of two tables $R$ and $S$ in time\linebreak $\Theta (n_R \log n_R + n_S
\log n_S + m \log m)$, where $n_R = |R|$, $n_S = |S|$, and $m = |R
\join S|$ using $O (\log (n_R + n_S))$ TM memory.
\end{theorem}

\subsection{Multiway Join}
\label{sec:nary-join-short}
We now consider multiway joins, i.e., natural joins between $q$
relations $R_1 \join \cdots \join R_q$. When the multiway join has a
property called \emph{acyclicity} there exists an efficient oblivious
algorithm for evaluating the join. 

The algorithm for evaluating a multiway join is a generalization of
the algorithm for binary join. Informally, we compute the contribution
of each tuple in $R_1, \ldots, R_q$ towards the final join. The
contribution generalizes the notion of a join-graph degree in the
binary join case, and this quantity can be computed by performing a
sequence of semi-join aggregations between the input relations. We
expand the input relations $R_1, \ldots, R_q$ to $R_{1,exp}, \ldots,
R_{q,exp}$ respectively by duplicating each tuple as many times as its
contribution, and stitch the expanded tables to produce the final join
output. The details of ordering the expansions $R_{1,exp}, \ldots,
R_{q,exp}$ are now more involved. A formal description of the overall
algorithm is deferred to the full-version~\cite{full-version}. Here we
present a formal characterization of the class of multiway join
queries our algorithm is able to handle.

\begin{definition}
\label{def:acyclic}
  The multiway join query $R_1 \join \cdots \join R_q$ is called
  \emph{acyclic}, if we can arrange the relations $R_1, \ldots, R_q$
  as nodes in a tree $T$ such that for all $i, j, k \in [1,q]$ such
  that $R_k$ is along the path from $R_i$ to $R_j$ in $T$,
  $\attrs(R_i) \isect \attrs(R_j) \subseteq \attrs (R_k)$.
\end{definition}

\begin{theorem}
\label{thm:obl-njoin-short}
There exists an oblivious algorithm to compute the natural join query
$(R_1 \join \cdots \join R_q)$ provided the query is acyclic. Further,
the time complexity of the algorithm is $\Theta(n \log n + m\log m)$
where $n = \sum_i |R_i|$ is the input size and $m = |R_1 \join \cdots
\join R_q|$ denotes the output size, and the TM memory requirement is
$O (\log (n+ m))$.
\end{theorem}
The concept of acyclicity is well-known in database
theory~\cite{DBLP:conf/vldb/Yannakakis81}; in fact, it represents the
class of multiway join queries for which algorithms polynomial in
input and output size are known. We use the acyclicity property to
compute the contribution of each tuple towards to the final output
using a series of semi-join aggregations. Without the acyclicity
property we do not know of a way of computing this quantity short of
evaluating the full join.

\subsection{Grouping and Aggregation}
\label{sec:groupagg}
This section presents an oblivious algorithm for grouping aggregation
over acyclic joins. We present our algorithm for the case of
$\mathrm{SUM}$; it can be easily adapted for the other standard
aggregation functions: $\mathrm{MIN}$, $\mathrm{MAX}$, $\mathrm{AVG}$,
and $\mathrm{COUNT}$. The algorithm handles only a limited form of
grouping where all the grouping attributes belong to a single
relation. The query is therefore of the form
$\groupagg_{\Gcal}^{\mathrm{SUM}(R_a.X)}$ $(R_1 \bowtie \cdots \bowtie
R_q)$, where (wlog) $\Gcal \in \attrs (R_1)$. We believe the case
where the grouping attributes come from multiple relations is hard as
we discuss in Section~\ref{sec:hardness}.

\vspace{1ex}\noindent\textbf{Notation:} Since the join $R_1 \join
\cdots \join R_q$ is acyclic we can arrange the relations as nodes in
a tree $T$ as per Definition~\ref{def:acyclic}. An algorithm for
constructing such a tree is presented in~\cite{conf/csac/YuO79}. For
the remainder of this section fix some tree $T$. Wlog, we assume that
relations $R_1, \ldots, R_q$ are numbered by a pre-order traversal of
tree $T$ so that if $R_i$ is an ancestor of $R_j$ then $i < j$. For
any relation $R_i$, we use $\#c(i)$ to denote the number of children
and $R_{c(i,1)}, \ldots, R_{c(i,\#c(i))}$ to denote the children of
$R_i$ in $T$; we denote the parent of $R_i$ using $R_{p(i)}$. We use
$\mathrm{Desc}(R_i)$ and $\mathrm{Anc}(R_i)$ to denote the descendants
and ancestors of $R_i$ in $T$; both $\mathrm{Anc}(R_i)$ and
$\mathrm{Desc}(R_i)$ contain $R_i$. For any set $\Rcal$ of relations
$\join \Rcal$ denotes the natural join of elements of $\Rcal$; e.g.,
$\join \mathrm{Desc}(R_i)$ denotes the join of $R_i$ and it
descendants.

Algorithm~\ref{alg:groupagg} presents our grouping aggregation
algorithm. The algorithm operates in $2$ stages: (1) a bottom-up
counting stage and (2) a grouping stage which works over just $R_1$.

\vspace{1ex}\noindent\textbf{Bottom-up Counting:} In this stage, we
add an attribute $N$ to each tuple. For a tuple $t \in R_i$, $t[N]$
denotes the number of join tuples $t$ is part of in $\join
\mathrm{Desc}(R_i)$. For leaf relations $R_i$, $t[N] = 1$ for all
tuples $t \in R_i$. For non-leaf relations, a simple recursion can be
used to compute the value of attribute $N$. Consider $t \in R_i$ for
some non-leaf $R_i$ and let $t[N_{c(i,j)}]$ denote the number of join
tuples $t$ is part of in the join $\join (\{R_i\} \union
\mathrm{Desc}(R_{c(i,j)})$ (join of all descendants rooted in child
$R_{c(i,j)}$). Then we can show using the acyclic property of the
join, $t[N] = \Pi_j t[N_{c(i,j)}]$ (Step~\ref{stp:ncomputegb}). In
addition, for all relations in $\mathrm{Anc}(R_a)$ ($R_a$ is the
relation containing aggregated column $X$), we add a partial
aggregation attribute $S_X$. For a tuple $t \in R_i \in
\mathrm{Anc}(R_a)$, $t[S_X]$ represents the sum of $R_a.X$ values in
$\join \mathrm{Desc}(R_i)$ considering only tuples that $t$ is part
of.

\vspace{1ex}\noindent\textbf{Grouping:} This stage essentially
performs the grouping $\groupagg_{\Gcal}^{\mathrm{SUM}(S_X)}
(R_1)$. We attach a unique id $\mathit{Id}_\Gcal$ to each tuple within
a group defined by $\Gcal$ (Step~\ref{stp:gb-gpid}). We then compute
the running sum of $S_X$ within each group defined by $\Gcal$; we
compute the running sum in descending order of $\mathit{Id}_\Gcal$ so
that the total sum for a group is stored with the record with
$\mathit{Id}_\Gcal = 1$. We get the final output by (obliviously)
selecting the records with $\mathit{Id}_\Gcal$ and performing suitable
projections (Step \ref{stp:gb-out}).

\begin{algorithm}[t]
  \small
  \caption{Grouping and Aggregation:
    $\groupagg_{\Gcal}^{\mathrm{SUM}(R_a.X)} (R_1 \bowtie \cdots
    \bowtie R_q)$}
  \label{alg:groupagg}
  \begin{algorithmic}[1]
    \Procedure{GroupingAggr}{$(R_1, \ldots, R_q), \Gcal, R_a.X$}
      \ForIter{$i=q$}{$1$} 
        \State $\Rtil_i \gets R_i$
        \If{$\#c(i) = 0$} 
          $\Rtil_i \gets \Rtil_i.(N \gets 1)$ \Comment{leaf table}
        \Else
          \ForIter{$j = 1$}{$\#c(i)$}
            \State $\Rtil_i \gets \Rtil_i.(N_{c(i,j)} \stackrel{\ltimes}{\gets}
            \mathit{Sum}(\Rtil_{c(i,j)}.N))$ 
          \EndForIter
          \State $\Rtil_i \gets \Rtil_i.(N \gets \Pi_{j=1}^{\#c(i)}
          N_{c(i,j)})$ \label{stp:ncomputegb}
        \EndIf
        \If{$R_i = R_a$} 
          \State $\Rtil_a \gets \Rtil_a.(S_X \gets X \times N)$ 
        \ElsIf{$R_a \in \mathrm{Desc}(R_i)$}
          \State $\ell \gets $ unique $\ell$ such that $R_a \in
          \mathrm{Desc}(R_{c(i,\ell)})$ 
          \State $\Rtil_i \gets \Rtil_i.(S_X \stackrel{\ltimes}{\gets} \mathit{Sum}(R_{c(i,\ell)}.S_X))$
        \EndIf
      \EndForIter
      \State $\Rtil_1 \gets \Rtil_1.(\mathit{Id}_\Gcal \gets
      \idfunc_{\Gcal})$ \label{stp:gb-gpid}
      \State $\Rtil_1 \gets \Rtil_1.(\mathit{RS}_X \gets
      \rsfunc_{\Gcal}^{-\mathit{Id}_\Gcal})$ \label{stp:gb-rs}
      \State Output $\project_{\Gcal, \mathit{RS}_X}
      (\select_{\mathit{Id}_\Gcal = 1} (\Rtil_1))$ \label{stp:gb-out}
    \EndProcedure
  \end{algorithmic}
\end{algorithm}

\begin{theorem}
\label{thm:groupagg}
  Algorithm~\ref{alg:groupagg} obliviously computes the grouping
  aggregation $\groupagg_{\Gcal}^{\mathrm{SUM}(R_a.X)} (R_1 \bowtie
  \cdots \bowtie R_q)$, $\Gcal \in \attrs(R_1)$, where $(R_1 \bowtie
  \cdots \bowtie R_q)$ is acyclic, in time $\Theta (n \log n)$ and using
  $O(\log n)$ TM memory, where $n = \sum_i |R_i|$ denotes the total
  size of the input relations.
\end{theorem}

\subsection{Selections}
\label{sec:filter-in-qry}

This section discusses how selections in SPJ and GSPJ queries can be
handled. We assume a selection predicate $P$ is a conjunction of
table-level predicates, i.e., of the form $(P_{R_{i1}} \wedge
P_{R_{i2}} \wedge \cdots)$ where each $P_{R_{ij}} : \domain{R_{i_j}}
\rightarrow \{ \mathrm{true}, \mathrm{false} \}$ is a binary predicate
over tuples of $R_{i_j}$. To handle selections in a GSPJ query
$\groupagg_\Gcal^{F(A)} (\select_P (R_1 \join \cdots \join R_q))$, we
modify the bottom-up counting stage of Algorithm~\ref{alg:groupagg} as
follows: when processing any relation $R_{i}$ with a table-level
predicate $P_{R_{i}}$ that is part of $P$, we set the value of
attribute $N$ to $0$ for all tuples $t_i \in R_{i}$ for which $P_{R_i}
(t_i) = \mathrm{false}$; for all other tuples the value of attribute
$N$ is calculated as before. We can show that the resulting algorithm
correctly and obliviously evaluates the query $\groupagg_\Gcal^{F(A)}
(\select_P (R_1 \join \cdots \join R_q))$. A similar modification
works for the SPJ query $\select_P (R_1 \join \cdots \join R_q)$ and
is described in the full version of the paper~\cite{full-version}.

\subsection{Exploiting Foreign Key Constraints}
\label{sec:fkey}
We informally discuss how we can exploit key-foreign key constraints;
We note that keys and foreign keys are part of database schema
(metadata), and we view them as public knowledge (see discussion in
Appendix~\ref{sec:query-sec}).  Consider a query $Q$ involving a
multiway join $R_1 \join \cdots \join R_q$ (with or without grouping)
and let $R_i$ (key side) and $R_j$ (foreign key side) denote two
relations involved in a key-foreign key join. We explicitly evaluate
$R_{ij} \gets R_i \join R_j$ using the oblivious binary join algorithm
and replace references to $R_i$ and $R_j$ in $Q$ with $R_{ij}$. From
key-foreign key property, it follows that $|R_{ij}| = |R_j|$, so this
step by itself does not render the query processing non-oblivious. We
treat any foreign key references to $R_j$ as references to
$R_{ij}$. We continue to this process of identifying key-foreign key
joins and evaluating them separately until no more such joins
exist. At this point, we revert to the general algorithms presented in
Sections~\ref{sec:nary-join-short}-\ref{sec:filter-in-qry} to process
the remainder of the query.

\begin{theorem} 
\label{thm:formal}
There exists an oblivious (secure) query processing algorithm that
requires $O(\log n)$ storage in TM for any database query involving
joins, grouping aggregation (as the outermost operation), and filters,
if (1) the non-foreign key join graph is acyclic and (2) all grouping
attributes are connected through foreign key joins, where $n$ denotes
the sum of query input and output sizes. Further, assuming no
auxiliary structures, the running time of the algorithm is within
$O(\log n)$ (multiplicative factor) of the running time of the best
insecure algorithm.
\end{theorem}
\begin{proof} (Sketch) From
  Theorems~\ref{thm:obl-binjoin}-\ref{thm:groupagg} and the informal
  descriptions in Section~\ref{sec:filter-in-qry} and \ref{sec:fkey},
  it follows that there exist oblivious query processing algorithms
  for above class of queries that run in $O(n \log n)$ time and
  require $O( \log n)$ TM memory. Further, any algorithm requires
  $\Omega (n)$ time without auxiliary structures.
\end{proof}

\begin{theorem}
\label{thm:ext-formal}
For the class of queries in Theorem~\ref{thm:formal} there exists an
oblivious (secure) algorithm with I/O complexity within multiplicative
factor $\log_{M/B} (n/B)$ of that of the optimal insecure algorithm,
where $B$ is the block size and $M$ is the TM memory.
\end{theorem}
\begin{proof} (Sketch) All of our algorithms are simple scans except
  for the steps that perform oblivious sorting. The I/O complexity of
  oblivious sorting is\linebreak $O (\frac{n}{B} \log_{M/B}
  (\frac{n}{B}))$~\cite{DBLP:conf/spaa/Goodrich11} from which the
  theorem follows.
\end{proof}


\section{Hardness Arguments}
\label{sec:hardness}

Section~\ref{sec:algorithm} presented efficient oblivious algorithms
for evaluating (G)SPJ queries when the underlying join was acyclic and
all grouping columns belonged to a single relation. All of our
algorithms have time complexity $O( (n + m) \log (n + m))$ where $n$
and $m$ and input and output sizes, respectively. Any algorithm
requires $\Omega (n + m)$ time without pre-processing, so our
algorithms are within a $\log (n + m)$ factor away from an instance
optimal algorithm. In the following, we call such oblivious algorithms
\emph{instance efficient}. While we do not have formal proofs, we
present some evidence that suggests that instance efficient oblivious
algorithms for cyclic joins and for the case where grouping columns
come from different tables seem unlikely to exist. 

At a high level, our arguments rely on the following intuition: if a
query $Q$ does not have a near-linear algorithm (with time complexity
$(n + m) \mathrm{polylog} (n + m)$) in the worst case it is unlikely
to have an oblivious algorithm since its behavior on an easy instance
would be different from that on a worst case instance. There are some
difficulties directly formalizing this intuition since some of the
computation occurs within TM potentially without an externally visible
data access.

Recent work has identified the following \emph{3SUM} problem as a
simple and useful problem for polynomial time lower-bound
reductions~\cite{DBLP:conf/stoc/Patrascu10}: Given an input set of $n$
numbers identify $x, y, z \in S$ such that $x + y = z$. There exists a
simple $O(n^2)$ algorithm for 3SUM: Store the $n$ numbers in $S$ in a
hashtable $\Hcal$. Consider all pairs of numbers $x, y \in S$ and
check if $x + y \in \Hcal$. It is widely believed that this algorithm
is the best possible and \cite{DBLP:conf/stoc/Patrascu10} uses this
3SUM-hardness conjecture to establish lower bounds for a variety of
combinatorial problems.

We introduce the following simple variant of 3SUM-hardness that
captures the additional complexity of TM computations. In this
variant, an algorithm has access to a \emph{cache} of size $n^\delta$
$(\delta < 1)$ words. Access to the cache is free while accesses to
non-cache memory have a unit (time) cost. We conjecture that having
access to a free cache does not bring down the asymptotic complexity
of 3SUM. For example, in the hashtable based solution above, only a
small part of the input (at most $n^\delta$) can be stored in the
cache and most of the hashtable lookups ($n - n^\delta$) incur a
non-cache access cost.
\begin{conjecture} 
  \label{conj:3sumcache} (3SUM-Cache($\delta$)-hardness) Any algorithm
  for 3SUM with input size $n$ having access to a free cache of size
  $n^\delta$ requires $\Omega(n^{2 - o(1)})$ time in expectation.
\end{conjecture}

Assuming this conjecture is true, the algorithms of
Section~\ref{sec:algorithm} almost represent a characterization of the
class of single-block queries that have an instance efficient
oblivious algorithm.

\begin{theorem}
\label{thm:cyc-join-hard}
There does not exist an instance efficient oblivious algorithm with a
TM with memory $n^{\delta}$ for cyclic joins over inputs of size $n$
unless there exists a subquadratic $O(n^{2 - \Omega(1)})$ algorithm
for 3SUM-Cache($\delta$).
\end{theorem}
\begin{proof}
  Enumerating $m$ triangles in a graph with $m$ edges in time
  $O(m^{4/3-\epsilon})$ is 3SUM-hard~\cite{DBLP:conf/stoc/Patrascu10}
  (Theorem 5). Enumerating triangles can be expressed as a cyclic join
  query over the edge relation. It follows that evaluating a cyclic
  join query in $O(m + n)^{4/3 - \epsilon}$ is 3SUM-hard. The free
  cache does not affect this reduction from 3SUM to cyclic join
  evaluation, so evaluating a cyclic join query using a TM with cache
  $n^\delta$ with $O(m + n)^{4/3 - \epsilon}$ UM memory accesses is
  3SUM-Cache($\delta$)-hard. We can construct easy instances for all
  sufficiently large $m, n$ that only require $O(m + n)$ processing
  time. It follows that there does not exist an instance efficient
  oblivious algorithm for cyclic joins.
\end{proof}

The \emph{set intersection enumeration} problem is the following:
given $k$ sets $S_1, \ldots, S_k$, $S_i \subseteq \Ucal$ drawn from
some universe $\Ucal$, identify all pairs $(i,j)$ such that $S_i
\isect S_j \neq \phi$. A simple algorithm is to build an inverted
index that stores for each element $e \in S_1 \union \ldots \union
S_k$ the list of integers $j$ such that $e \in S_j$. We consider all
pairs of integers within each list and output the pair if we have not
already done so. This simple algorithm is quadratic in the input and
output sizes in the worst case. The set intersection enumeration
problem is fairly well-studied and is at the core of most
high-dimensional~\cite{DBLP:conf/webdb/HaveliwalaGI00} and approximate
string matching~\cite{DBLP:conf/vldb/ArasuGK06} but no asymptotically
better algorithm is known. There is a simple reduction from set
intersection enumeration to evaluating grouping queries where grouping
attributes come from multiple relations.

\begin{theorem}
\label{thm:sie-red}
  If there exists a $O((n + m) \mathrm{polylog}(n + m))$ time
  algorithm for evaluating $\groupagg_{A,B} (R (A, Id) \join S (B,
  Id))$ where $n$ and $m$ are input and output sizes of the query,
  then there exists an $O((n + m) \mathrm{polylog}(n + m))$ time
  algorithm for set intersection enumeration.
\end{theorem}
\begin{proof}
  We encode the input to set intersection enumeration as two relations
  $R (A, Id)$ and $S(B, Id)$. The domain of $A$ and $B$ is $\Ucal$ the
  universe of elements. For each $e \in S_i$ we include a tuple $(e,
  i)$ in both $R$ and $S$. ($R$ and $S$ are therefore identical.) The
  theorem follows from the observation that the desired output of set
  intersection enumeration is precisely $\groupagg_{A,B} (R (A, Id)
  \join S (B, Id))$.
\end{proof}
Theorem~\ref{thm:sie-red} along with the fact that we can construct an
input instance for the query $\groupagg_{A,B} (R (A, Id) \join S (B,
Id))$ which can be evaluated in $O((n + m) \mathrm{polylog}(n + m))$
time suggests that an oblivious algorithm for this query is likely to
imply a (significantly) better algorithm for set intersection
enumeration than currently known.


\section{Related Work}
\label{sec:rel}

While we focused mostly on systems that use trusted hardware for
querying encrypted data, there exist other approaches. One such
approach relies on \emph{homomorphic encryption} that allows
computation directly over encrypted data; e.g., the \emph{Paillier
  cryptosystem} \cite{Paillier99} allows us to compute the encryption
of $(v_1 + v_2)$ given the encryptions of $v_1$ and $v_2$ without
requiring the (private) encryption key, and can be used to process
\texttt{SUM} aggregation queries~\cite{GeZ07}. However, despite recent
advances, practical homomorphic encryption schemes are currently known
only for limited classes of computation. There exist simple queries
that the state-of-the-art systems~\cite{PopaRZB11} that rely solely on
homomorphic encryption cannot process. A second approach is to use the
client as the trusted location inaccessible to the
adversary~\cite{HacigumusILM02, monomi}. One drawback of using the
client is that some queries might necessitate moving large amounts of
data to the client for query processing and defeat the very purpose of
using a cloud service. We can reduce some of the above drawbacks by
combining homomorphic encryption with the client processing
approach~\cite{monomi}, but, given the current state-of-the-art of
homomorphic encryption, a comprehensive and robust solution to
querying encrypted data seems to require the trusted hardware
architecture we have assumed in this paper.


\section{Acknowledgments}
Discussions with the Cipherbase team greatly helped us at all stages
of this work. We thank Chris Re, Atri Rudra, and Hung Ngo for pointing
us to hardness of enumerating triangles result from
\cite{DBLP:conf/stoc/Patrascu10}, and Bryan Parno and Paraschos
Koutris for valuable feedback on initial drafts. We also thank Dan
Suciu and anonymous reviewers for helping better position our work
against ORAM simulation.


\bibliographystyle{plain}
\bibliography{security}
\pagebreak
\appendix
\normalsize
\section{Query and Metadata Security}
\label{sec:query-sec}

Ideally a secure query processing system hides both database contents
and queries being evaluated against the database. Given how databases
are typically used, we believe it is important to study the security
of database contents even if the adversary has full knowledge of the
query. Databases are typically accessed using a front end application,
and the entropy of the queries run over the database is typically
small given knowledge of the application. As a concrete example,
consider a paper review system like EasyChair, which is a web
application that presumably stores and queries papers and review
information using a backend database. The web application might itself
be well known (or open source~\cite{DBLP:conf/nsdi/Kohler08}), so the
kinds of queries issued to the database is public knowledge.

That said, the following question remains: To what extent should a
secure query processing system hide queries? What should an adversary
without knowledge of the query be able to learn? We argue that perfect
query security---being able to mask whether a query is a single table
filter or a 50-table join---is impractical given the richness of
database query languages. Going the extra mile to hide queries is also
wasteful if we assume knowledge of the application is easy to acquire
as in the discussion of the paper review system above.

What an application (source code) does not usually specify is actual
query constants which might be generated, e.g., by users filling in forms
or from values in the database and it might be valuable (and as it turns
out, easy) to secure these. Accordingly, in this paper we equate query
security with \emph{query constants} security: an adversary might
learn by observing a query execution, that e.g., the query is a 2-way
join with a filter on the first table, but she does not learn the
filter predicate constants. We note that this is either
explicitly~\cite{HacigumusILM02,PopaRZB11,monomi} or
implicitly~\cite{cipherbase,BajajS11} the notion of query security
used in prior work.

A related issue is metadata security. Metadata refers to information
such as number of tables and the schema (column names and types) of each
table. While column and table names can be easily anonymized, for the
same reasons mentioned in query security, formally securing metadata
is both difficult and not very useful if the adversary has knowledge
of the application.


\end{document}